\documentclass[12pt]{article}
\newtheorem{proposition}{Proposition}

\newtheorem{lemma}{Lemma}

\newenvironment{proof}{\noindent{\bf Proof:}}{\hfill\fbox{}\vspace*{1mm}}
\usepackage[usenames]{color}
\usepackage{graphicx}
\begin{document}
\title{\bf On Infectious Model for Dependent Defaults}
\author{Jia-Wen Gu
\thanks{Advanced Modeling and Applied Computing Laboratory,
Department of Mathematics, The University of Hong Kong,
Pokfulam Road, Hong Kong. Email:jwgu.hku@gmail.com.
}
\and Wai-Ki Ching
\thanks{Corresponding author. Advanced Modeling and Applied Computing Laboratory,
Department of Mathematics, The University of Hong Kong, Pokfulam
Road, Hong Kong. E-mail: wching@hku.hk. Research supported in
part by RGC Grants 7017/07P, HKU CRCG Grants
and HKU Strategic Research Theme Fund on Computational Physics and Numerical Methods.}
\and Tak-Kuen Siu
\thanks{ Department of Applied Finance and Actuarial Studies,
Faculty of Business and Economics, Macquarie University,
Macquarie University, Sydney, NSW 2109, Australia. Email: ken.siu@mq.edu.au,
ktksiu2005@gmail.com,}
\and Harry Zheng
\thanks{ Department of Mathematics,
Imperial College, London, SW7 2AZ, UK. Email: h.zheng@imperial.ac.uk.}
}
\date{}
\maketitle

\begin{abstract}

In this paper, we propose a two-sector Markovian infectious model, which is an extension of Greenwood's model.
The central idea of this model is that the causality of defaults of two sectors is in both direction, which enrich dependence dynamics.
The Bayesian Information Criterion is adopted to compare the proposed model with the two-sector model in credit literature using the real data.
We find that the newly proposed model is statistically better than the model in past literature.
We also introduce two measures: CRES and CRVaR to give risk evaluation of our model.

\end{abstract}


\section{Introduction}

Modeling dependent default risk has been a key issue in credit risk modeling.
There are two important approaches to model the dependent default risk.
The structural firm model has its origin in Merton (1974) and Black and Scholes (1973), which models the relationship between the firm's asset value and the defaults.
The reduced-form intensity-based model by Jarrow and Turnbull (1995) use Poisson jump processes to model the default event.

Copula has been a very popular tool in modeling the dependent risk.
The idea of Copula is transforming the marginal variables to uniform variables by a simple transformation.
After this is done, a n-dimensional function is used to model the dependence of the uniform variables,
which is so called the Copula function.
The Copula helps us to deal with the multivariate distribution of the uniform variable, without consideration of the
original marginal variables.
There are many useful Copulas in finance. The Gaussian Copula, which is introduced by Li (2000),
is widely used in risk modeling and financial assessment.

In addition, conditional independence model is also a commonly used model in credit risk modeling. Conditional on the systematical common factor, the loss random variables are independent. To specify, the Bernoulli mixture model is followed by the $CreditMetrics$ and $KMV$-model, while the Poisson mixture model is followed by the $CreditRisk^+$ model.
In a recession, the default of one company is triggered by the underlying common risk factor and also by the related company's defaults.
The contagion model is used to describe how the credit event of one company affects the other companies.
Davis and Lo (2001) introduce an infectious default model, where in a portfolio a bond may be infected by defaults of other bonds or default directly.
Jarrow and Yu (2001)  propose a reduced-form model to describe the defaultable bonds of different company, where the concept of counterparty risk is first introduced to the credit literature.

Ching et al. (2008) introduce an infectious default model based on
the idea of Greenwood's model considered in Daley and Gani (1999) .
This model aims at modeling the impact of
default of a bond on the likelihood of defaults of other bonds.
The original version of Greenwood's model is a one-sector model.
It is then extended to a two-sector model in Ching et al. (2008).
Besides, the joint probability distribution function
for the duration of a default crisis, (ie, the default cycle),
and the severity of defaults during the crisis period
was also derived. Two concepts, namely, the Crisis Value-at-Risk (CRVaR)
and the Crisis Expected Shortfall (CRES), are also used
to assess the impact of a default crisis.
The Greenwood's model is also extended to a network of sectors in
Ching et al. (2010). Gu et al. (2011) propose a Markovian infectious
model to describe the dependent relationship of default processes of
credit securities based on Ching et al. (2008, 2010),
where the central idea is the concept of
common shocks which is one of the major approaches to
describe insurance risk.

In this paper, we propose a two-sector Markovian infectious model, where
the future default probability switching over time depends on the
current number of defaults of both sectors.
Moreover, the defaults of sector A caused by th defaults of sector B, and vice versa.
The causality of defaults in both direction is captured by the underlying switched default probability.
We adopt the maximum likelihood method to estimate the parameters and the Bayesian Information Criterion to compare the propose model with two-sector model considered in Ching et al. (2008). The experiment result shows that the proposed model outperforms the model in credit literature.
{In addition, a more general model
is given to provide more flexibility in describing
realistic features of the dynamics of default probabilities.}

This paper is structured as follows.
Section 2 presents our proposed model.
And we also derive a recursive formula for the joint probability
distribution for the default cycle and the number of defaults during
the crisis and outline the estimation procedure.
Section 3 presents the ideas of the CRVaR and the CRES.
In Section 4, we present the results of empirical analysis using our proposed model.
{Section 5 gives the general model extending the proposed model
in Section 2. The final section concludes the paper.}
\vspace{5mm}

\section{The Basic Model}

Let ${\cal T}$ be the time index set $\{0, 1, 2, \ldots, \}$ of our model.
To model the uncertainty, we consider a probability space $(\Omega, {\cal F}, {\cal P})$, where ${\cal P}$
is a real-world probability. Suppose that
$$
X := \{ {X_t} \}_{t \in {\cal T}} \quad {\rm and} \quad Y := \{ {Y_t} \}_{t \in {\cal T}}
$$
denote two stochastic processes on $(\Omega, {\cal F}, {\cal P})$, where $X_t = ( {X^1_t}, {X^2_t} )$
and $Y_t = ( {Y^1_t}, {Y^2_t} )$ represent the numbers of surviving bonds and the defaulted bonds at
$t \in {\cal T}$ in sector A and sector B, respectively, e.g., $X^1_t$ represents the the number of surviving bonds
at time t in sector A.
We assume that the initial conditions are given as follows:
$$
X_0 = (x^1_0 , x^2_0) ,\quad Y_0 = (y^1_0 , y^2_0) \quad {\rm and} \quad
x^1_0 + y^1_0 = N_1 ,\quad x^2_0 + y^2_0 = N_2 \ .
$$
Note that for each $t \in {\cal T}$, the sum of the numbers of the
defaulted bonds and the surviving bonds at the time epoch $t+1$ must
equal the number of surviving bonds at time $t$ in every sector, i.e.,
\begin{equation}\label{xy}
X^1_{t+1}+Y^1_{t+1} = X^1_t \quad {\rm and} \quad  X^2_{t+1}+Y^2_{t+1} = X^2_t .
\end{equation}
For each $t\in {\cal T}$,
let ${\alpha}_t$ and ${\beta}_t$ be the probability that the default of
a surviving bond is infected by the defaulted bonds at time $t$ in sector A and sector B, respectively.
The joint probability distribution
of $\{X_{t+1}, Y_{t+1}\}$ given $\{X_t, Y_t\}$ is given by
the following Binomial probability:
\begin{eqnarray}\label{prob}
\begin{array}{lll}
p_{(x_t,y_t)} (x_{t+1},y_{t+1})  & =&
P\{(X_{t+1}, Y_{t+1}) = (x_{t+1}, y_{t+1}) \mid (X_t, Y_t) = (x_t, y_t)\} \\
&=&  \left(
\begin{array}{c}
x^1_t \\
y^1_{t+1}
\end{array}
\right) {(\alpha}_t)^{y^1_{t+1}} (1-{\alpha}_t)^{x^1_{t+1}}
\times \left(
\begin{array}{c}
x^2_t \\
y^2_{t+1}
\end{array}
\right) {(\beta}_t)^{y^2_{t+1}} (1-{\beta}_t)^{x^2_{t+1}}.\\
\end{array}
\end{eqnarray}
We consider here the situation that the joint future default probability
depends on the current number of defaulted bonds of both industrial sectors.
We assume that
\begin{eqnarray}
\begin{array}{llll}
{\alpha}_t &=& a(y_t) \\
&=&\left\{
\begin{array}{llll}
a_0 & {\rm if} & y^1_t=y^2_t=0\\
a_1 & {\rm if} & y^1_t>0, y^2_t=0\\
a_2 & {\rm if} & y^1_t=0, y^2_t>0\\
a_3 & {\rm if} & y^1_t>0, y^2_t>0
\end{array}
\right.\\
&=& a_0h_0(y^1_t, y^2_t)+a_1h_1(y^1_t, y^2_t)+a_2h_2(y^1_t, y^2_t)+a_3h_3(y^1_t, y^2_t)
\end{array}
\end{eqnarray}
and
\begin{eqnarray}
\begin{array}{lll}
{\beta}_t &=& b(y_t) \\
&=&\left\{
\begin{array}{llll}
b_0 & {\rm if} & y^1_t=y^2_t=0\\
b_1 & {\rm if} & y^1_t=0, y^2_t>0\\
b_2 & {\rm if} & y^1_t>0, y^2_t=0\\
b_3 & {\rm if} & y^1_t>0, y^2_t>0
\end{array}
\right.\\
&=&b_0h_0(y^2_t,y^1_t)+b_1h_1(y^2_t, y^1_t)+b_2h_2(y^2_t, y^1_t)+b_3h_3(y^2_t, y^1_t),
\end{array}
\end{eqnarray}
where
$$
h_0(x,y)=\left\{
\begin{array}{ll}
1 \ & {\rm if} \ x=y=0\\
0 \ & {\rm otherwise}
\end{array}
\right.
, \quad
h_1(x,y)=\left\{
\begin{array}{ll}
1 \ & {\rm if} \ x>0, y=0\\
0 \ & {\rm otherwise}
\end{array}
\right.
$$
and
$$
h_2(x,y)=\left\{
\begin{array}{ll}
1 \ & {\rm if} \ x=0, y>0\\
0 \ & {\rm otherwise}
\end{array}
\right.
 , \quad
h_3(x,y)=\left\{
\begin{array}{ll}
1 \ & {\rm if} \ x>0, y>0\\
0 \ & {\rm otherwise}.
\end{array}
\right.
$$

As it is shown in Equation ({\ref {xy}})({\ref {prob}}), one can see that $\{X_t, t = 0, 1, 2, \ldots\}$ is a second-order Markov chain process.
{
We remark that this two-sector model provides a novel and flexible dependent structure for
correlated defaults of two different industrial sectors.
Firstly, an infectious default within one time period is modeled as a
Binomial distribution, which has been widely used in modeling the spread
of epidemics whose situation is quite similar to that of a financial crisis.
The causality of the infection is supposed to be in both direction, i.e.,
a ``looping default''.
Secondly, the process $(X_t, Y_t)$ has the Markov property,
where the probabistic structure of future states only depend
on the current state. Thirdly, conditioning on the current state $(X_t, Y_t)$,
the future state of two sectors $(X^1_{t+1}, Y^1_{t+1})$
and $(X^2_{t+1}, Y^2_{t+1})$ are stochastically independent.
The step functions $h_i(x,y)$ are used to describe the dependence
of the default probabilities on the state of previous time epoch.
On one hand, this method provides a tractable and analytic solution
for parameter estimation from empirical data.
On the other hand, one has to admit that this simplicity may
result in limitations in applications.
In Section 5, we relax the assumption of
the specific form for $\alpha_t$ and $\beta_t$ and
a more complicated dependent structure modelling framework
is presented.}

\subsection{Default Cycle and Severity}
In this subsection, we proceed to derive the joint probability distribution function (p.d.f) for the duration of the default crisis ($T$), namely, the default cycle, and the severity of the defaults ($W_T$) during the crisis period.
{ These two concepts are essential in determining the impact of a default crisis.
We first give a precise definition of the default cycle:}
\begin{equation}
T := {\rm inf}\{t \in {\cal T} \mid Y_t = 0\}.
\end{equation}
And given $T = t > 0$, $W_t$ represents the number of defaults in the sector over the time duration $(0, t]$.
To apply the concepts of default cycle and the severity of the defaults on our proposed two-sector model, we write
$$
T_1 := {\rm inf}\{t \in {\cal T} \mid Y^1_t = 0\} \quad {\rm and} \quad T_2 := {\rm inf}\{t \in {\cal T} \mid Y^2_t = 0\}.
$$
{
Provided that $T_1=t_1>0$ and $T_2=t_2>0$, $W^1_{t_1}$ and $W^2_{t_2}$
represent the number of defaults in sector A and sector B respectively in
$(0, t_1]$ and $(0, t_2]$.
To obtain the joint distribution of $(W^i_{T_i}, T_i)$ for $i=1, 2$,
we assume that $(X_0, Y_0) =(x_0, y_0)$ with $y^1_0>0, y^2_0>0$.
Let
$$
P_n(x_1,x_2,h)=P\{T_1 \geq n+1, X^1_n=x_1, X^2_n=x_2, I_{\{Y^2_n >0\}}=h\}.
$$
The following Lemma gives recursive formulas for $P_n(x_1,x_2,h)$.
\begin{lemma}
$$
\begin{array}{ll}
P_n(x_1,x_2,0)=& \displaystyle \sum_{s_1>x_1}{s_1 \choose x_1} \left[ P_{n-1}(s_1,x_2,0)(a_1)^{s_1-x_1}(1-a_1)^{x_1}(1-b_2)^{x_2}\right.\\
&\displaystyle \left.+P_{n-1}(s_1,x_2,1)(a_3)^{s_1-x_1}(1-a_3)^{x_1}(1-b_3)^{x_2} \right]\\
 P_n(x_1,x_2,1)=& \displaystyle \sum_{s_1>x_1}\sum_{s_1>x_1}{s_1 \choose x_1} {s_2 \choose x_2}\left[ P_{n-1}(s_1,s_2,0)(a_1)^{s_1-x_1}(1-a_1)^{x_1}(b_2)^{s_2-x_2}(1-b_2)^{x_2}\right.\\
&\displaystyle \left.+P_{n-1}(s_1,s_2,1)(a_3)^{s_1-x_1}(1-a_3)^{x_1}(b_3)^{s_2-x_2}(1-b_3)^{x_2} \right]
\end{array}
$$
where the initial condition is given by
$$
P_0(x_1,x_2,h)=\left\{
\begin{array}{ll}
1, & (x_1,x_2,h)=(x^1_0,x^2_0,1)\\
0, & {\rm otherwise}.
\end{array}
\right.
$$
\end{lemma}

\begin{proof}
By the law of total probability and Markov property,
$$
\begin{array}{lll}
 & P_n(x_1,x_2,0)\\
=& P\{T_1 \geq n+1, X^1_n=x_1, X^2_n=x_2, I_{\{Y^2_n >0\}}=0\}\\
=& {\displaystyle \sum_{s_1>x_1}\sum_{h=0,1}}P\{T_1 \geq n, X^1_{n-1}=s_1, X^2_{n-1}=x_2, I_{\{Y^2_{n-1} >0\}}=h\}\\
& \times P\{T_1 \geq n+1, X^1_n=x_1, X^2_n=x_2, I_{\{Y^2_n >0\}}=0 \mid T_1 \geq n, X^1_{n-1}=s_1, X^2_{n-1}=x_2, I_{\{Y^2_{n-1} >0\}}=h\}\\
=&{\displaystyle \sum_{s_1>x_1}\sum_{h=0,1}} P_{n-1}(s_1,x_2,h)\\
 & \times P\{Y^1_n>0, X^1_n=x_1, X^2_n=x_2, I_{\{Y^2_n >0\}}=0 \mid T_1 \geq n, X^1_{n-1}=s_1, X^2_{n-1}=x_2, I_{\{Y^2_{n-1} >0\}}=h\}\\
 =&{\displaystyle \sum_{s_1>x_1}\sum_{h=0,1}} P_{n-1}(s_1,x_2,h)\\
 & \times P\{Y^1_n>0, X^1_n=x_1, X^2_n=x_2, I_{\{Y^2_n >0\}}=0 \mid Y^1_{n-1}>0, X^1_{n-1}=s_1, X^2_{n-1}=x_2, I_{\{Y^2_{n-1} >0\}}=h\}\\
=&{\displaystyle \sum_{s_1>x_1}} {s_1 \choose x_1} \left[ P_{n-1}(s_1,x_2,0)(a_1)^{s_1-x_1}(1-a_1)^{x_1}(1-b_2)^{x_2}\right.\\
&\displaystyle \left.+P_{n-1}(s_1,x_2,1)(a_3)^{s_1-x_1}(1-a_3)^{x_1}(1-b_3)^{x_2} \right]
\end{array}
$$
Similarly, we have
$$
\begin{array}{lll}
 &P_n(x_1,x_2,1)\\
=& P\{T_1 \geq n+1, X^1_n=x_1, X^2_n=x_2, I_{\{Y^2_n >0\}}=1\}\\
=& {\displaystyle \sum_{s_1>x_1}\sum_{s_2>x_2}\sum_{h=0,1}}P\{T_1 \geq n, X^1_{n-1}=s_1, X^2_{n-1}=s_2, I_{\{Y^2_{n-1} >0\}}=h\}\\
& \times P\{T_1 \geq n+1, X^1_n=x_1, X^2_n=x_2, I_{\{Y^2_n >0\}}=1 \mid T_1 \geq n, X^1_{n-1}=s_1, X^2_{n-1}=s_2, I_{\{Y^2_{n-1} >0\}}=h\}\\
=&{\displaystyle \sum_{s_1>x_1}\sum_{s_2>x_2}\sum_{h=0,1}} P_{n-1}(s_1,s_2,h)\\
 & \times P\{Y^1_n>0, X^1_n=x_1, X^2_n=x_2, I_{\{Y^2_n >0\}}=1 \mid T_1 \geq n, X^1_{n-1}=s_1, X^2_{n-1}=s_2, I_{\{Y^2_{n-1} >0\}}=h\}\\
 =&{\displaystyle \sum_{s_1>x_1}\sum_{s_2>x_2}\sum_{h=0,1}} P_{n-1}(s_1,s_2,h)\\
 & \times P\{Y^1_n>0, X^1_n=x_1, X^2_n=x_2, I_{\{Y^2_n >0\}}=1 \mid Y^1_{n-1}>0, X^1_{n-1}=s_1, X^2_{n-1}=s_2, I_{\{Y^2_{n-1} >0\}}=h\}\\
=&{\displaystyle \sum_{s_1>x_1}\sum_{s_1>x_1}} {s_1 \choose x_1} {s_2 \choose x_2}\left[ P_{n-1}(s_1,s_2,0)(a_1)^{s_1-x_1}(1-a_1)^{x_1}(b_2)^{s_2-x_2}(1-b_2)^{x_2}\right.\\
& \displaystyle \left.+P_{n-1}(s_1,s_2,1)(a_3)^{s_1-x_1}(1-a_3)^{x_1}(b_3)^{s_2-x_2}(1-b_3)^{x_2} \right]
\end{array}
$$

\end{proof}
\begin{proposition}
The joint distribution of $(T_1,W^1_{T_1} )$ is given by
$$
P\{(T_1, W^1_{T_1})=(n,x)\}=\sum_{x_2}P_{n-1}(x^1_0-x, x_2,0)(1-a_1)^{x^1_0-x}+\sum_{x_2}P_{n-1}(x^1_0-x, x_2,1)(1-a_3)^{x^1_0-x}.
$$
\end{proposition}

\begin{proof}
$$
\begin{array}{lll}
 &P\{(T_1,W^1_{T_1})=(n,x)\}\\
=&P\{T_1 \geq n, Y^1_n=0, X^1_n=x^1_0-x\}\\
=& {\displaystyle \sum_{x_2}\sum_{h=0,1}}P\{T_1 \geq n, X^1_{n-1}=x^1_0-x, X^2_{n-1}=x_2, I_{\{Y^2_{n-1} >0\}}=h\}\\
& \times P\{ Y^1_n=0, X^1_n=x^1_0-x \mid T_1 \geq n, X^1_{n-1}=x^1_0-x, X^2_{n-1}=x_2, I_{\{Y^2_{n-1} >0\}}=h\}\\
=& {\displaystyle \sum_{x_2}\sum_{h=0,1}}P_{n-1}(x^1_0-x,x_2,h)\\
& \times P\{ Y^1_n=0, X^1_n=x^1_0-x \mid Y^1_{n-1}>0, X^1_{n-1}=x^1_0-x, X^2_{n-1}=x_2, I_{\{Y^2_{n-1} >0\}}=h\}\\
=&\sum_{x_2}P_{n-1}(x^1_0-x, x_2,0)(1-a_1)^{x^1_0-x}+\sum_{x_2}P_{n-1}(x^1_0-x, x_2,1)(1-a_3)^{x^1_0-x}
\end{array}
$$
\end{proof}

We remark that due to the symmetric property of the two sectors, the joint distribution $(W^2_{T_2}, T_2)$ shares a similar form of $(W^1_{T_1}, T_1)$.
}
\subsection{Parameter Estimation}
This two-sector model has eight parameters: $a_0$, $a_1$, $a_2$, $a_3$ and $b_0$, $b_1$, $b_2$, $b_3$. We employ the maximum likelihood method to estimate the parameters. Given the total bonds $N_1$, $N_2$ and the observations of the number of defaulted bonds $y^1_0, y^1_1, \ldots, y^1_N$ and
$y^2_0, y^2_1, \ldots, y^2_N$, where $N$ denotes the period of observation time, the number of surviving binds $x^1_0, x^1_1, \ldots, x^1_N$ and $x^2_0, x^2_1, \ldots, x^2_N$ are deterministic.
The following proposition gives analytical expressions for the maximum likelihood estimates of the model parameters.
\begin{proposition}
For $i=0, 1, 2, 3$,
$$
\hat{a}_i=\frac{\sum\limits_{t=0}^{N-1}y^1_{t+1}h_i(y^1_t, y^2_t)}{\sum\limits_{t=0}^{N-1}x^1_{t}h_i(y^1_t, y^2_t)}
\quad {\rm and} \quad
\hat{b}_i=\frac{\sum\limits_{t=0}^{N-1}y^2_{t+1}h_i(y^2_t, y^1_t)}{\sum\limits_{t=0}^{N-1}x^2_{t}h_i(y^2_t, y^1_t)}
.
$$
\end{proposition}

\begin{proof}
We prove the expression for $\hat{a}_0$ here 
and the proof for the others are similar.
The likelihood function $L(a,b \mid x_0, x_1, \ldots, x_N, y_0, y_1, \dots, y_N)$ is then the joint probability density function
$f(x_0, x_1, \ldots, x_N, y_0, y_1, \ldots, y_N \mid a,b)$:
$$
\begin{array}{llllll}
&&L(a, b \mid x_0, x_1, \ldots, x_N, y_0, y_1, \ldots, y_N)\\
&=& f(x_0, x_1, \ldots, x_N, y_0, y_1, \ldots, y_N \mid a, b)\\
&=&\left(
\begin{array}{c}
x^1_0\\
x^1_1
\end{array}
\right)
(1-a(y_0))^{x^1_1}{a(y_0)}^{y^1_1}
\times\left(
\begin{array}{c}
x^2_0\\
x^2_1
\end{array}
\right)
(1-b(y_0))^{x^2_1}{b(y_0)}^{y^2_1}\\
&&\times\left(
\begin{array}{c}
x^1_1\\
x^1_2
\end{array}
\right)
(1-a(y_1))^{x^1_2}{a(y_1)}^{y^1_2}
\times\left(
\begin{array}{c}
x^2_1\\
x^2_2
\end{array}
\right)
(1-b(y_1))^{x^2_2}{b(y_1)}^{y^2_2}
\ldots\ldots\\
&&\times\left(
\begin{array}{c}
x^1_{N-1}\\
x^1_{N}
\end{array}
\right)
(1-a(y_{N-1}))^{x^1_{N}}{a(y_{N-1})}^{y^1_N}
\times\left(
\begin{array}{c}
x^2_{N-1}\\
x^2_N
\end{array}
\right)
(1-b(y_{N-1}))^{x^2_N}{b(y_{N-1})}^{y^2_N}.
\end{array}
$$
Then by solving
$$
\frac{\partial \ln L(a, b \mid x_0,x_1,\ldots,x_N, y_0,y_1,\ldots,y_N )}{\partial a_0} = 0 \ ,
$$
we have
$$
-\sum\limits_{t=0}^{N-1}\frac{x^1_{t+1}h_0(y^1_t,y^2_t)}{1-a(y_t)}+\sum\limits_{t=0}^{N-1}\frac{y^1_{t+1}h_0(y^1_t,y^2_t)}{a(y_t)}=0.
$$
Since for any $t$,
$$
\frac{1}{1-a(y_t)}=\sum\limits_{i=0}^{3}\frac{h_i(y^1_t,y^2_t)}{1-a_i}
\quad {\rm and} \quad
\frac{1}{a(y_t)}=\sum\limits_{i=0}^{3}\frac{h_i(y^1_t,y^2_t)}{a_i} ,
$$
then
$$
\begin{array}{lll}
0&=&
\displaystyle -\sum\limits_{t=0}^{N-1}\sum\limits_{i=0}^{3}
\frac{x^1_{t+1}h_0(y^1_t,y^2_t)h_i(y^1_t,y^2_t)}{1-a_i}+
\sum\limits_{t=0}^{N-1}\sum\limits_{i=0}^{3}
\frac{y^1_{t+1}h_0(y^1_t,y^2_t)h_i(y^1_t,y^2_t)}{a_i}\\
&=& \displaystyle -\sum\limits_{t=0}^{N-1}
\frac{x^1_{t+1}h_0(y^1_t,y^2_t)}{1-a_0}+
\sum\limits_{t=0}^{N-1}
\frac{y^1_{t+1}h_0(y^1_t,y^2_t)}{a_0}.
\end{array}
$$
Thus,
$$
\hat{a}_0=\frac{\sum\limits_{t=0}^{N-1}y^1_{t+1}h_0(y^1_t, y^2_t)}{\sum\limits_{t=0}^{N-1}x^1_{t}h_0(y^1_t, y^2_t)}
$$
\end{proof}

\section{Crisis VaR and Crisis ES}

In this section, we give a brief introduction to the concepts of the CRVaR and the CRES in Ching et al. (2010).
Then we present the evaluation of the CRVaR and the CRES using the proposed models.
The CRVaR and the CRES are measures for the duration and the severity of a default crisis.
Let
$$
L (\cdot, \cdot) (\omega) : {\cal T} \times {\cal R} \times \Omega \rightarrow {\cal R}
$$
be a real-valued function $L (T, W_T) (\omega)$ of $T$ and $W_T$.
We then suppose that for a fixed $\omega \in \Omega$,
$$
T(\omega) = t, \quad W_t (\omega) = w, \quad {\rm and} \quad L (t, w) (\omega) = l (t, w) \in {\cal R}.
$$
That is, the loss from the default crisis is $l (t, w)$ when the
duration of default crisis $T = t$ and the number of defaulted bonds
in the crisis $W_t = w$.
We write ${\cal L} (T, W_T)$ for the space of all loss functions
$L (T, W_T) (\omega)$ generated by $T$ and $W_T$.

The CRVaR with probability level $\beta$ under ${\cal P}$ is then defined as
a functional $V_{\beta} (\cdot) : {\cal L} (T, W_T) \rightarrow {\cal R}$ such that
for each $L (T, W_T) \in {\cal L} (T, W_T)$,
\begin{eqnarray}
V_{\beta} (L (T, W_T)) := \inf \{ l \in {\cal R} | {\cal P} (L (T, W_T) > l )  \le \beta \} \ .
\end{eqnarray}
In the language of statistics, $V_{\beta} (L (T, W_T))$ is the
generalized $\beta$-quantile of the distribution of the loss
variable $L (T, W_T)$ under ${\cal P}$. Since the loss from the
default crisis $L(T, W_T)$ is completely determined when $T$ and
$W_T$ are given, ${\cal P} (L (T, W_T) > l )$ is completely
determined by the joint p.d.f. of $W_T$ and $T$.

The CRES with probability level $\beta$ under ${\cal P}$ is also defined as
a functional $E_{\beta} (\cdot) : {\cal L} (T, W_T) \rightarrow {\cal R}$ such that
for each $L (T, W_T) \in {\cal L} (T, W_T)$,
\begin{eqnarray}
E_{\beta} (L (T, W_T) ) := E_{\cal P} [ L (T, W_T) | L (T, W_T) \ge V_{\beta} (L (T, W_T)) ].
\end{eqnarray}
In other words, $E_{\beta} (L (T, W_T) )$ is the average of the loss from the
default crisis when the loss exceeds the CRVaR of the default crisis with
probability level $\beta$ under ${\cal P}$.
\vspace{5mm}

\section{Empirical Results for Proposed Model}

In this section we present the empirical results of the proposed two-sector model using real default data extracted from
the figures in Giampieri et al. (2005), where we adopt the estimation methods and techniques presented in the previous section.

The default data comes from four different sectors.
They include consumer/service sector, energy and natural resources sector,
leisure time/media sector and transportation sector.
Table 1 shows the default data taken from Giampieri et al. (2005).
From the table, the proportions of defaults for Consumer, Energy, Media and Transport
are $24.1\%$, $16.9\%$, $20.5\%$ and $21.0\%$, respectively.
The default probabilities of all four sectors are significantly greater than zero.
This means that the default risk of each of the four sectors is substantial.
\begin{table}
\centering
\begin{tabular}{|c|c|c|}
\hline
{\bf Sectors}   & {\bf Total} & {\bf Defaults}  \\
\hline
{\bf Consumer} & 1041  & 251 \\
{\bf Energy}   & 420   & 71  \\
{\bf Media}    & 650   & 133 \\
{\bf Transport}& 281   & 59  \\
\hline
\end{tabular}
\caption{The default data (Taken from Giampieri et al. (2005)).}
\end{table}

We then construct the infectious disease model using these real data.
The asterisk ``*'' in the table indicates the pair of sectors
which has the largest correlation. From Table 2, we see that all correlations are positive.
This provides some preliminary evidence for supporting
the use of the two-sector model from the perspective of
descriptive statistical analysis.
We shall provide more empirical evidence for supporting the use of the proposed infectious
model by the results of Bayesian Information Criterion (BIC) later in this section.
To build the infectious model, for each row (sector A), we may find a
partner (sector B) by searching the one with the largest correlation
in magnitude (ie, the one with the asterisk ``*'').
Figure 1 gives the partner relations among the sectors using correlation.
Later in this section, we will give the results for BIC to support the matched pair presented in figure 1.
The estimation results for proposed infectious model and two-sector model Ching et al. (2010) are presented in Table 3.

\begin{figure}
\centering
\begin{picture}(200,200)(50,0)
\put(60, 35){Transport} \put(235, 35){Media}
\put(75, 170){Energy} \put(245, 170){Consumer}
\put(235,38){\vector(-1,0){120}}
\put(260,165){\vector(0,-1){120}}
\put(260,45){\vector(0,1){120}}
\put(245,45){\vector(-1,1){120}}
\end{picture}
\caption{The partner relations among the sectors using correlation.}
\end{figure}
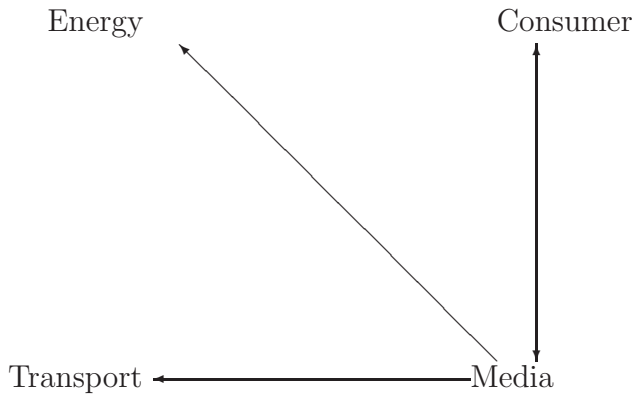

\begin{table}
\begin{center}
\begin{tabular}{|c|c|c|c|c|}
\hline
        & {\bf Consumer} & {\bf Energy} & {\bf Media} & {\bf Transport} \\
\hline
{\bf Consumer} & -          & 0.0224 &0.6013$^*$ &0.3487\\
{\bf Energy}   & 0.0224     & -      &0.1258$^*$ &0.1045\\
{\bf Media}    & 0.6013$^*$ &0.1258  &-          &0.3708\\
{\bf Transport}& 0.3487     &0.1045  & 0.3708$^*$&-\\
\hline
\end{tabular}
\end{center}
\caption{Correlations of the sectors.}
\end{table}

\begin{table}[h]
\centering
\begin{tabular}{|c|c|c|c|c|}
  \hline
  \textbf{Sector A}  & \textbf{Consumer}   & \textbf{Energy}  & \textbf{Media}   & \textbf{Transport}  \\
  \textbf{Sector B}  & \textbf{Media}     & \textbf{Media}   & \textbf{Consumer} & \textbf{Media}\\
  \hline \hline
  \multicolumn{5}{|c|}{Proposed Model}\\
  \hline
  $a_0$         & 0.0007      & 0.0004  & 0.0005  & 0.0013 \\
  $a_1$        & 0.0018       & 0.0033  & 0.0005   & 0.0012 \\
  $a_2$         & 0.0013      & 0.0018  & 0.0017   & 0.0026\\
  $a_3$        & 0.0049       & 0.0032  & 0.0042   & 0.0052 \\
  \hline \hline
  \multicolumn{5}{|c|}{Two-sector Model \cite{Ching1}}\\
  \hline
$\alpha_0$   & 0.0013   &0.0018       &0.0005   &0.0013\\
$\alpha_1$   & 0.0043   &0.0023       &0.0033   &0.0036\\
\hline
\end{tabular}
\caption{Estimation Results for Proposed Model}
\end{table}

To compare the proposed infectious model with the two-sector model Ching et al. (2010), we consider the Bayesian information criterion (BIC).The formula for the BIC is
$$
{\rm BIC}=-2log(L)+klog(n),
$$
where $n$ is the number of observation data, $k$ is the number of free parameters to be estimated, and $L$ is the maximized value of the likelihood function for the estimated model. Given any two estimated models, the model with the lower value of BIC is the one to be preferred. Table 4 presents the value of the BIC for the proposed model and the two-sector Ching et al. (2010). We remark that for all the four sectors, the proposed model with lower value of BIC is statistically better.
\begin{table}[h]
\centering
\begin{tabular}{|c|c|c|c|c|}
  \hline
  \textbf{Sector A}   & \textbf{Consumer} & \textbf{Energy} & \textbf{Media}    & \textbf{Transport} \\
    \textbf{Sector B}   & \textbf{Media}    & \textbf{Media}  & \textbf{Consumer} & \textbf{Media} \\
   \hline \hline
   BIC(proposed model)                     & 419.0813     & 215.4654      & 301.2534      &  2.1287  \\
  \hline
   BIC(two-sector model Ching et al. (2010))     & 434.6700     & 231.8225      & 321.0501       &  2.1460  \\
  \hline
\end{tabular}
\caption{The Value of BIC for Proposed Model and Two-sector Model Ching et al. (2010)}
\end{table}

To compare the matched pairs in Figure 1 with other matched pairs for the proposed model, we also adopt the BIC.
Since the models of different matched pairs have the same number of parameters and length of data set, to compare their BIC is equivalent to
compare their log-likelihood ratio.
Table 4 presents the log-likelihood ratios for the matched pairs in Figure 1 against other matched pairs.
We remark that all the log-likelihood ratios are positive which support the matched pairs in Figure 1 for the proposed model.

\begin{table}
\centering
\begin{tabular}{|c|c|c|c|c|}
\hline
\multicolumn{5}{|c|}{Matched Pairs in Figure 1}\\
\hline
   \textbf{Sector A}   & \textbf{Consumer} & \textbf{Energy} & \textbf{Media}    & \textbf{Transport} \\
  \textbf{Sector B}   & \textbf{Media}    & \textbf{Media}  & \textbf{Consumer} & \textbf{Media} \\
  \hline \hline
\multicolumn{5}{|c|}{Other Matched Pairs} \\
\hline
  \textbf{Sector A}   & \textbf{Consumer} & \textbf{Energy} & \textbf{Media}    & \textbf{Transport} \\
  \textbf{Sector B}   & \textbf{Energy}    & \textbf{Consumer}  & \textbf{Energy} & \textbf{Consumer} \\
   \hline
  log-likelihood ratio   & 33.1330     & 7.3286      & 18.6264      &  1.9942  \\
  \hline
  \textbf{Sector A}   & \textbf{Consumer}  & \textbf{Energy}     & \textbf{Media}      & \textbf{Transport} \\
  \textbf{Sector B}   & \textbf{Transport} & \textbf{Transport}  & \textbf{Transport}  & \textbf{Energy} \\
   \hline
  log-likelihood ratio   & 10.7231     & 7.3495      & 14.6136      &  8.4934  \\
  \hline
\end{tabular}
\caption{The Value of BIC for Matched pairs in Figure 1 and Other Matched Pairs}
\end{table}

Our proposed model aims at modeling causality of defaults in both direction.
From the pair up results, one may found that the relation is not necessarily symmetric.
This relation is only found symmetric for the sectors media and consumer, \
which means the causality of defaults from both direction is more reasonable for the media and consumer sector.

We provide a scatter plot to depict the correlation of defaults in 
the matched sectors. A simulation of defaults in matched sectors
in our proposed model is also conducted.
Figure 2 presents the number of surviving bonds in 
the matched sectors of empirical data and simulation.

\begin{figure}[h]
\centering
		\resizebox{7cm}{7cm}{\includegraphics{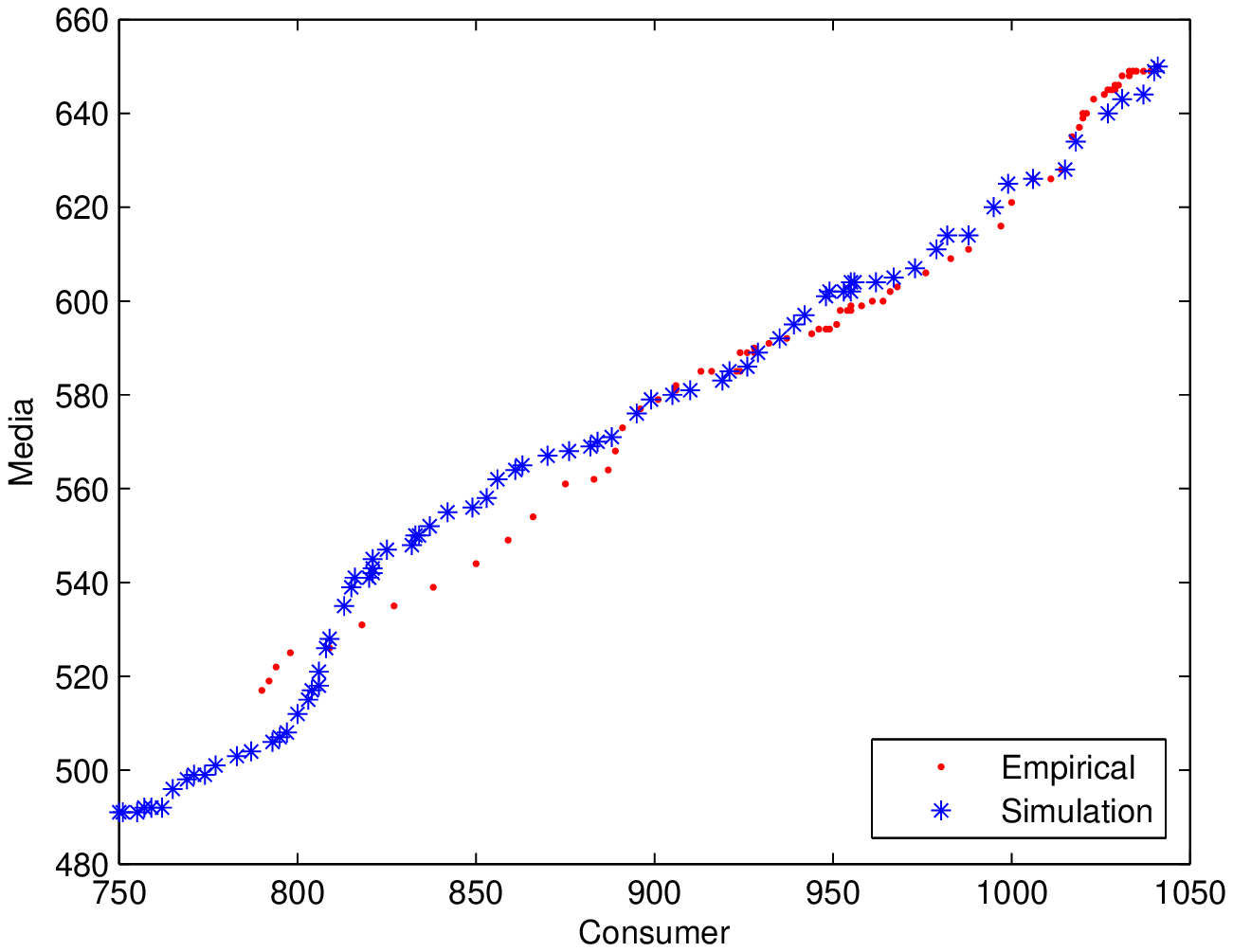}}
		\resizebox{7cm}{7cm}{\includegraphics{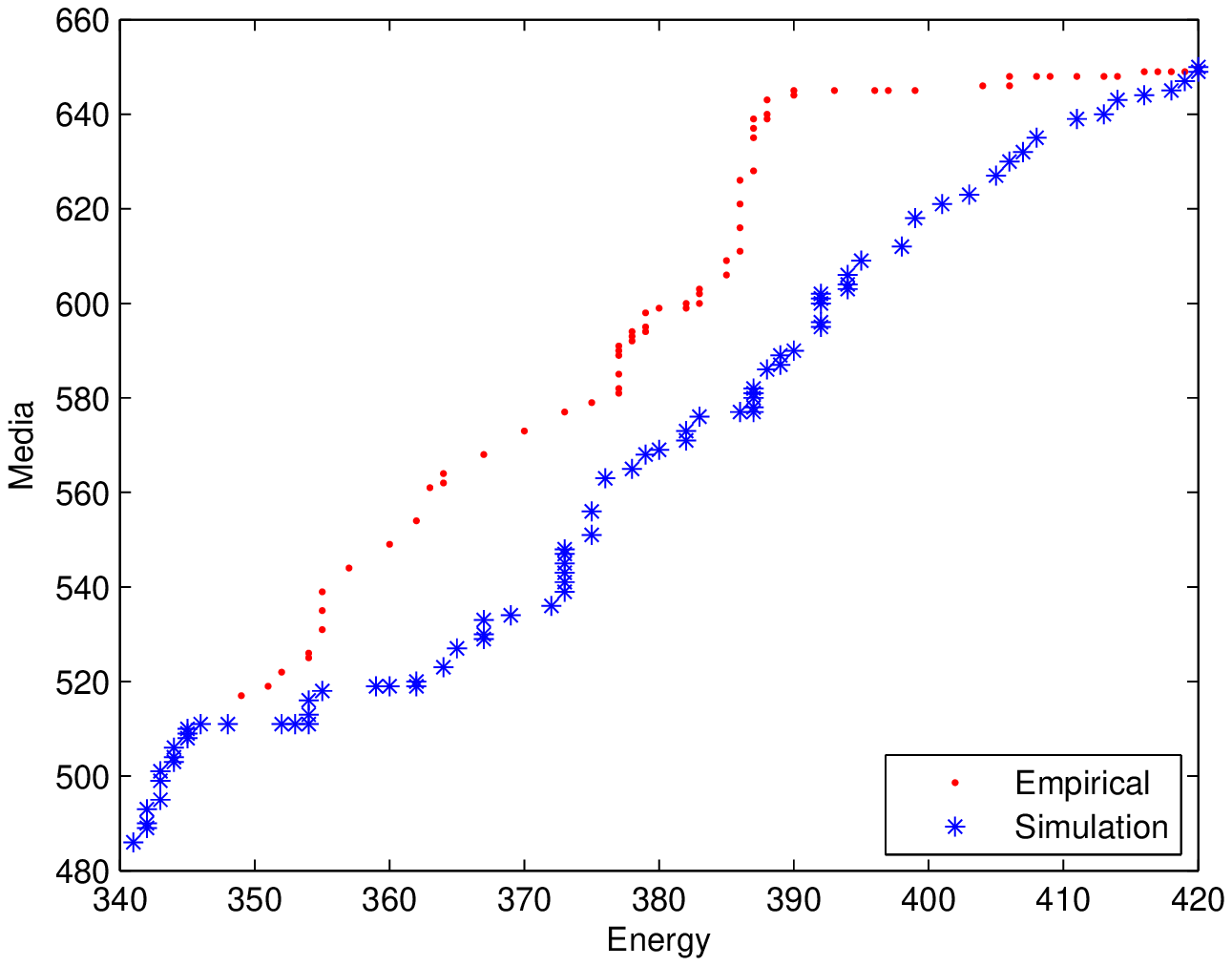}}\\
		\resizebox{7cm}{7cm}{\includegraphics{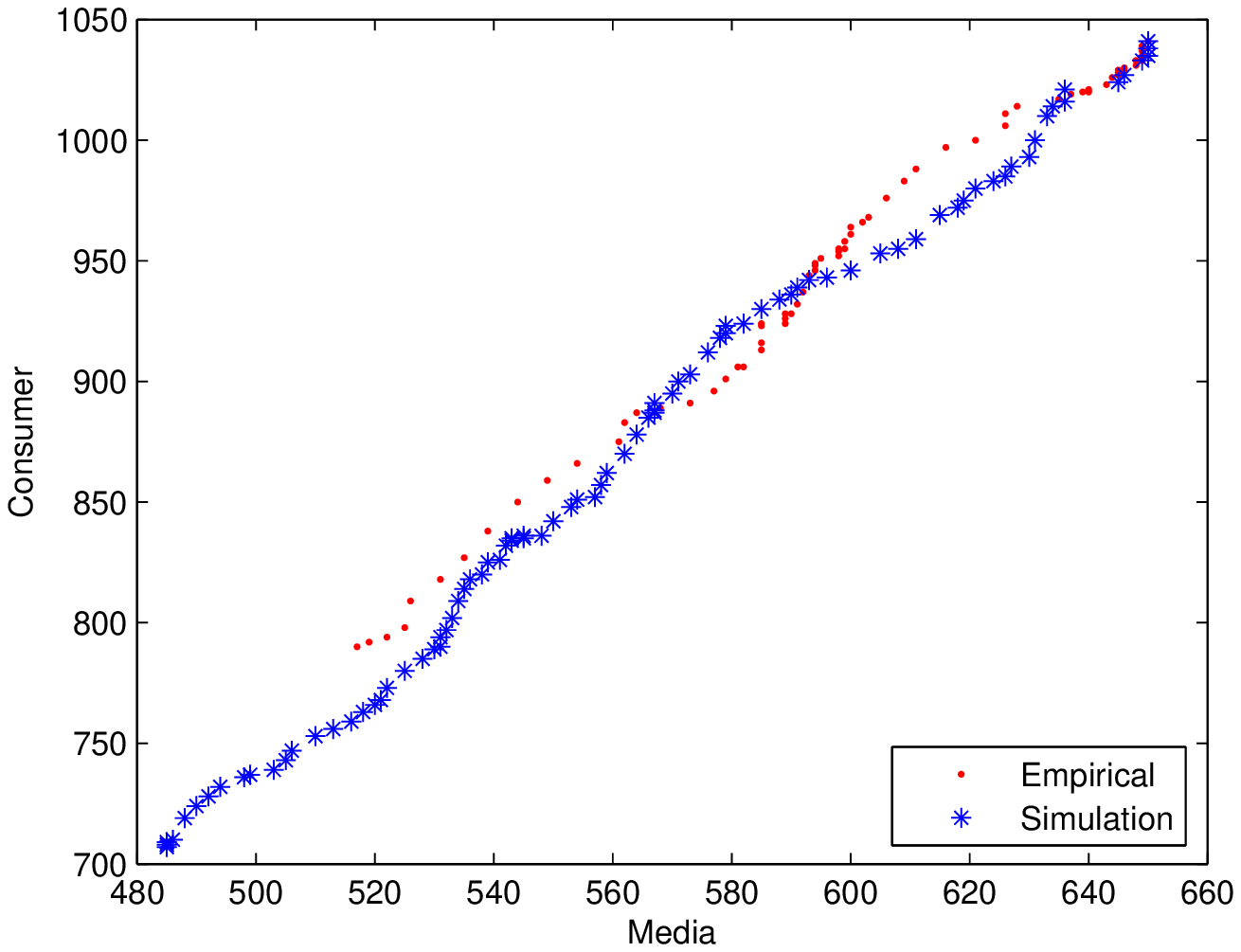}}
		\resizebox{7cm}{7cm}{\includegraphics{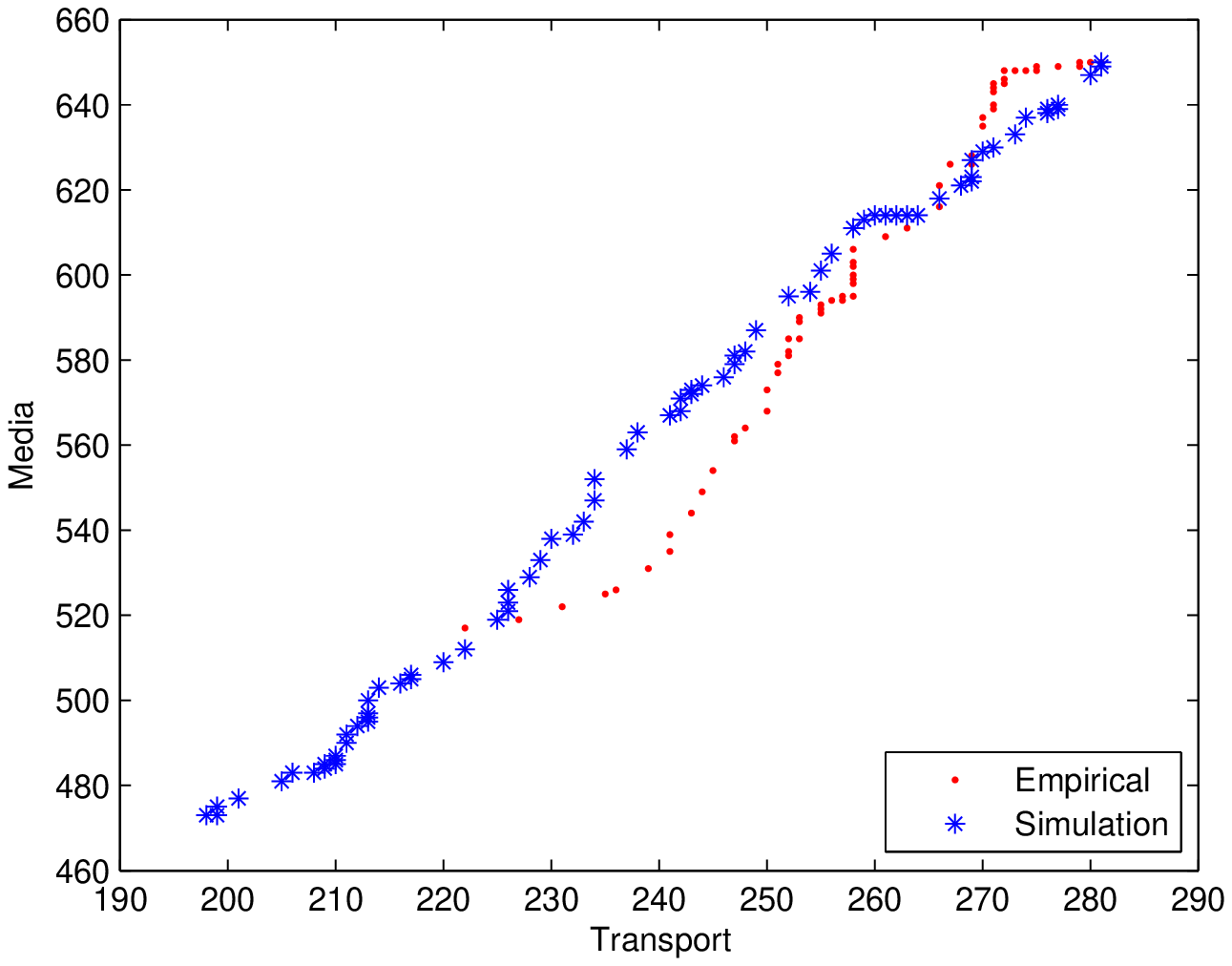}}
	\caption{Number of Surviving Bonds in Matched Sectors.}
\end{figure}
\vspace{5mm}

To apply the two measures CRVaR and CRES in the proposed model,
we consider some hypothetical values for the loss.
The loss $L(W_T,T )$, for each $T = 1, 2, \ldots, X_0$ and $W_T = 0, 1, \ldots, X_0$, are as in (\ref {loss}).
Then we present the value of CRVaR and CRES for the proposed model as well as the two-sector model Ching et al. (2010) in Table 6.
And the loss distribution are presented in figure 3.
\begin{eqnarray}\label{loss}
\left\{
\begin{array}{l}
L(0,j) = j - 1 + 0.1, \ {\rm for \ each} \ j = 1, 2, \ldots, X_0;\\
L(i,j) = L(0,j) + i - 1, \ {\rm for \ each} \ i = 1, 2, \ldots, X_0 \ {\rm and} \ j = 1, \dots, X_0.\\
\end{array}
\right.
\end{eqnarray}

\begin{table}[h]
\centering
\begin{tabular}{|c|c|c|c|c|}
 \hline
  \textbf{Sector A}   & \textbf{Consumer}   & \textbf{Energy}  & \textbf{Media}     & \textbf{Transport}   \\
  \textbf{Sector B}   & \textbf{Media}      & \textbf{Media}   & \textbf{Consumer}  & \textbf{Media} \\
  \hline \hline
  \multicolumn{5}{|c|}{Proposed Model}\\
  \hline
  CRVaR$(\beta=0.05)$           &374.1     &  25.1     & 122.1   &  26.1  \\
  CRES$(\beta=0.05)$            &424.7     & 33.8      & 150.4   & 33.8   \\
 \hline
  CRVaR$(\beta=0.01)$           &457.1     &  39.1     & 168.1   & 39.1   \\
  CRES$(\beta=0.01)$            &495.1     & 47.5      &192.4    &  46.5  \\
 \hline \hline
 \multicolumn{5}{|c|}{Two-sector Model Ching et al. (2010)}\\
 \hline
  CRVaR$(\beta=0.05)$           & 114.1    &  12.1     & 34.1   &  10.10  \\
  CRES$(\beta=0.05)$            & 146.1   & 17.1      & 45.7   & 14.1   \\
 \hline
  CRVaR$(\beta=0.01)$           &166.1    &  20.1    & 52.1   & 16.1   \\
  CRES$(\beta=0.01)$            & 195.6   & 24.5      &63.3    &  20.2  \\
 \hline
\end{tabular}
\caption{CRVaR and CRES}
\end{table}

\begin{figure}[h]
\centering
		\resizebox{7cm}{7cm}{\includegraphics{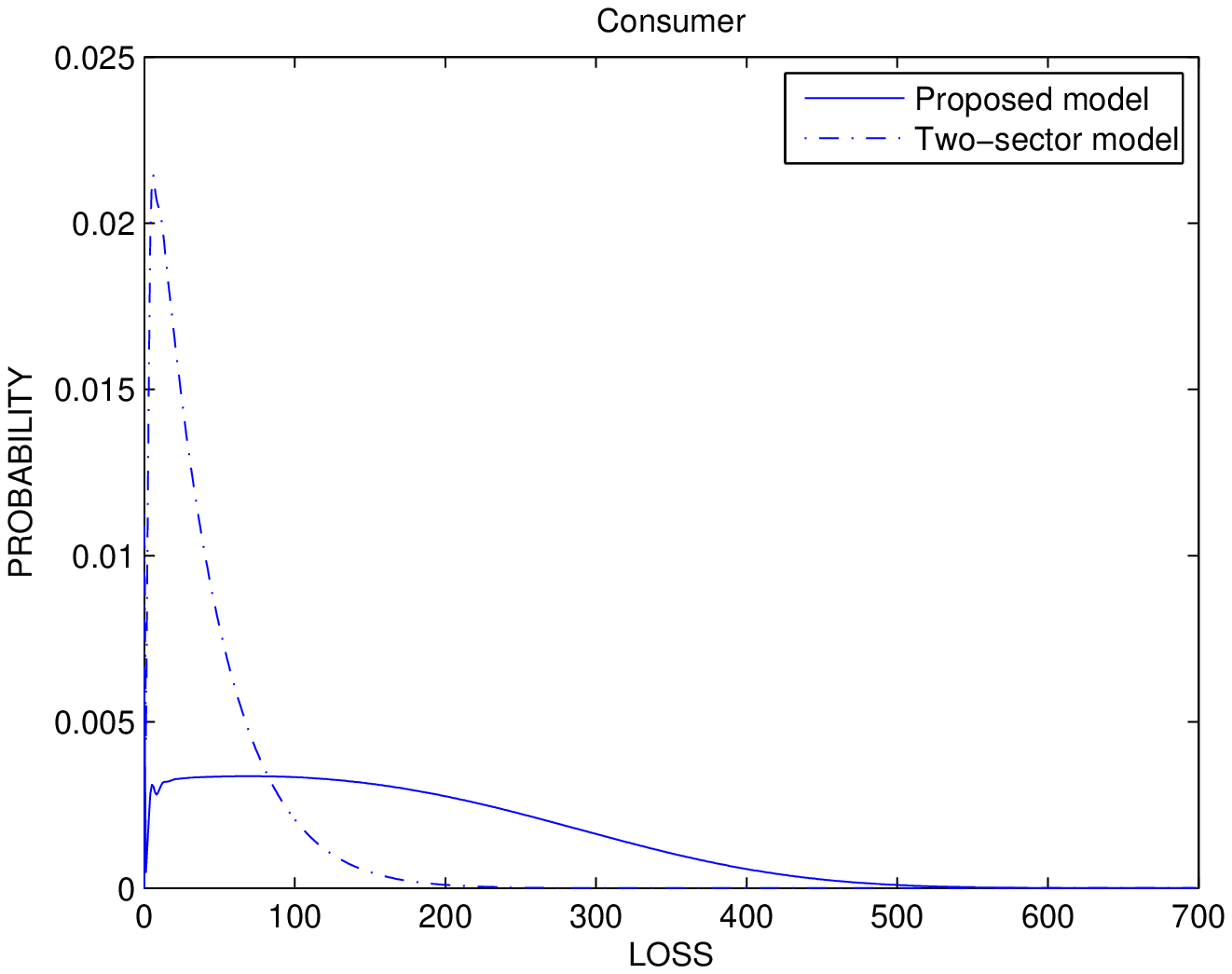}}
		\resizebox{7cm}{7cm}{\includegraphics{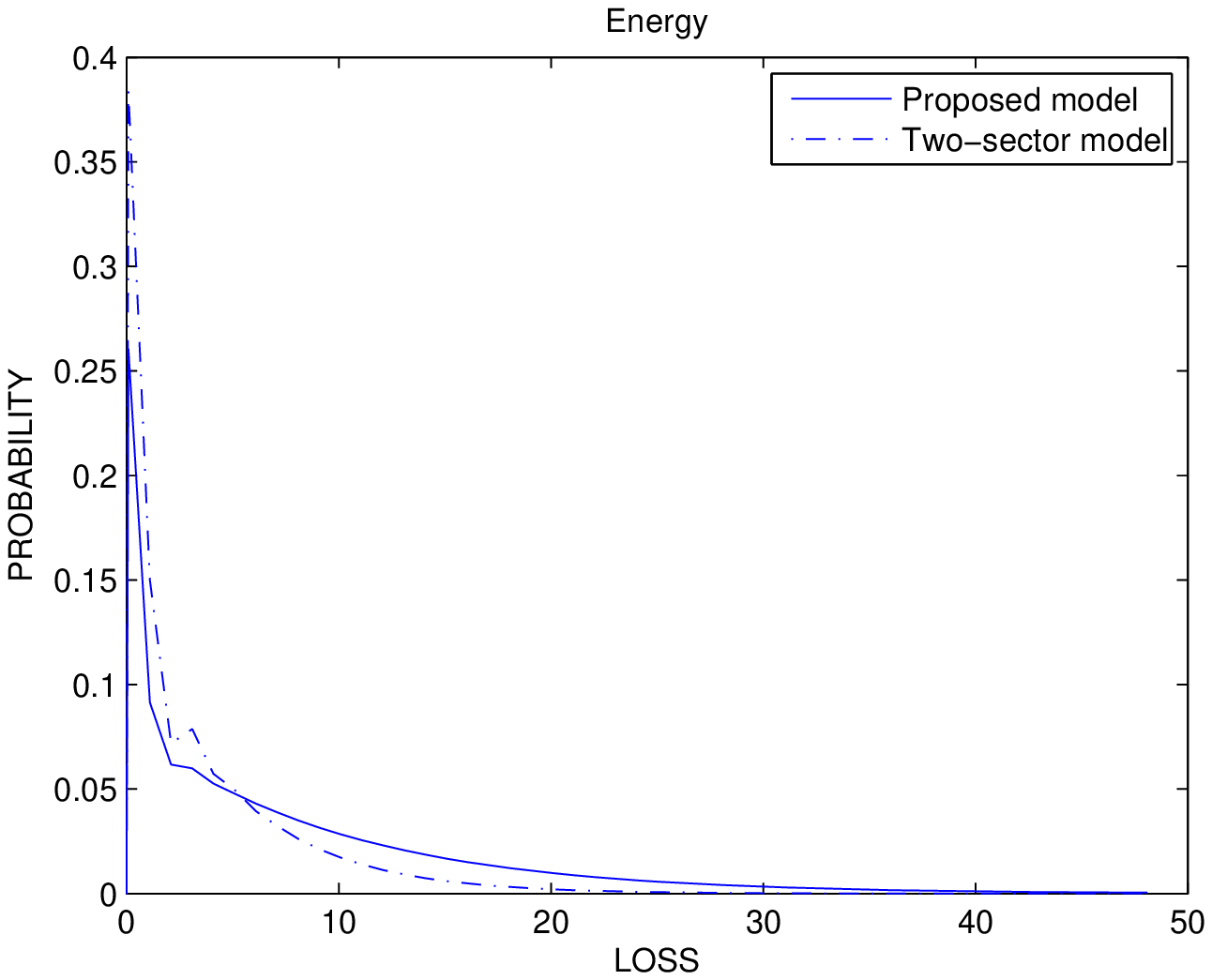}}\\
		\resizebox{7cm}{7cm}{\includegraphics{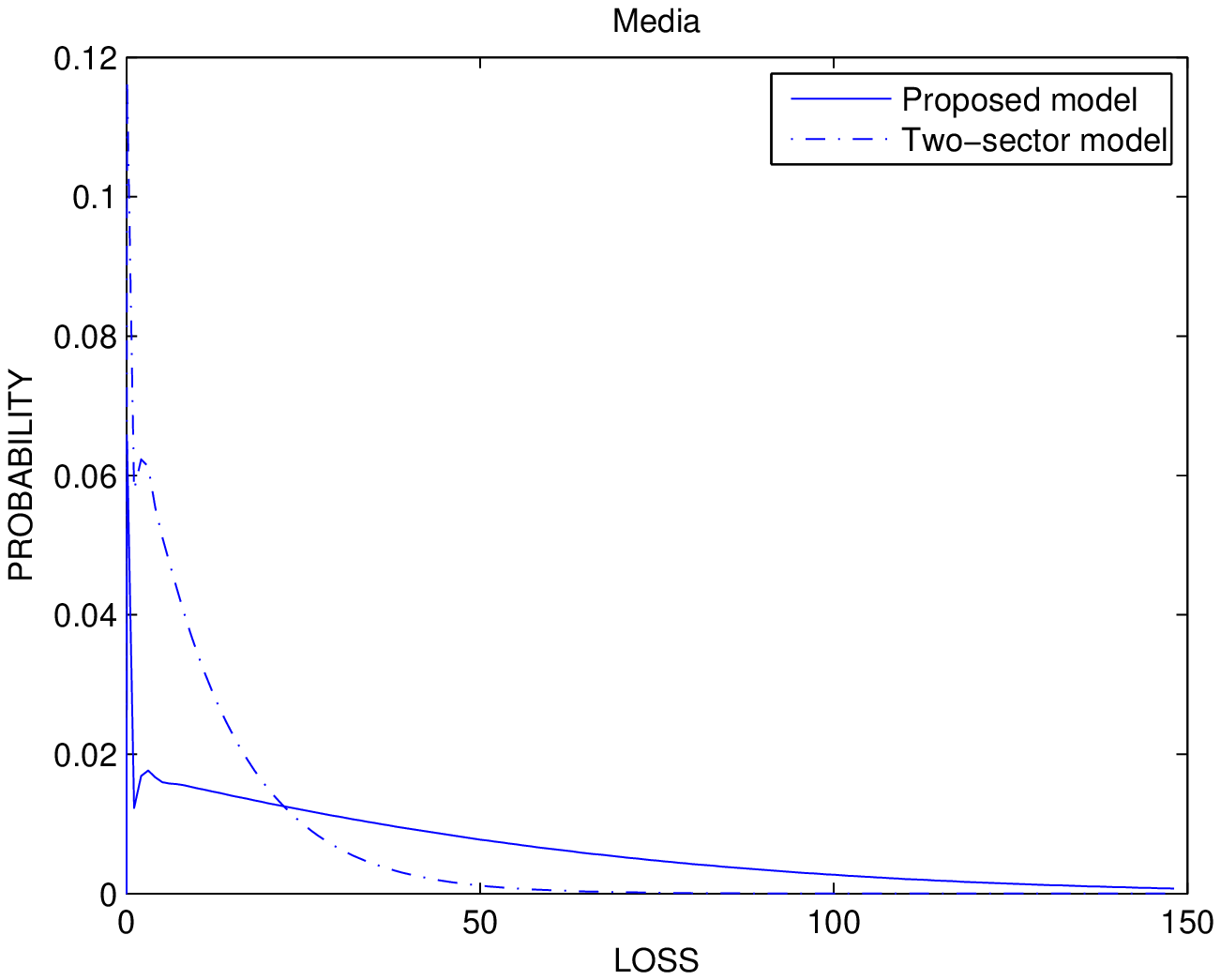}}
		\resizebox{7cm}{7cm}{\includegraphics{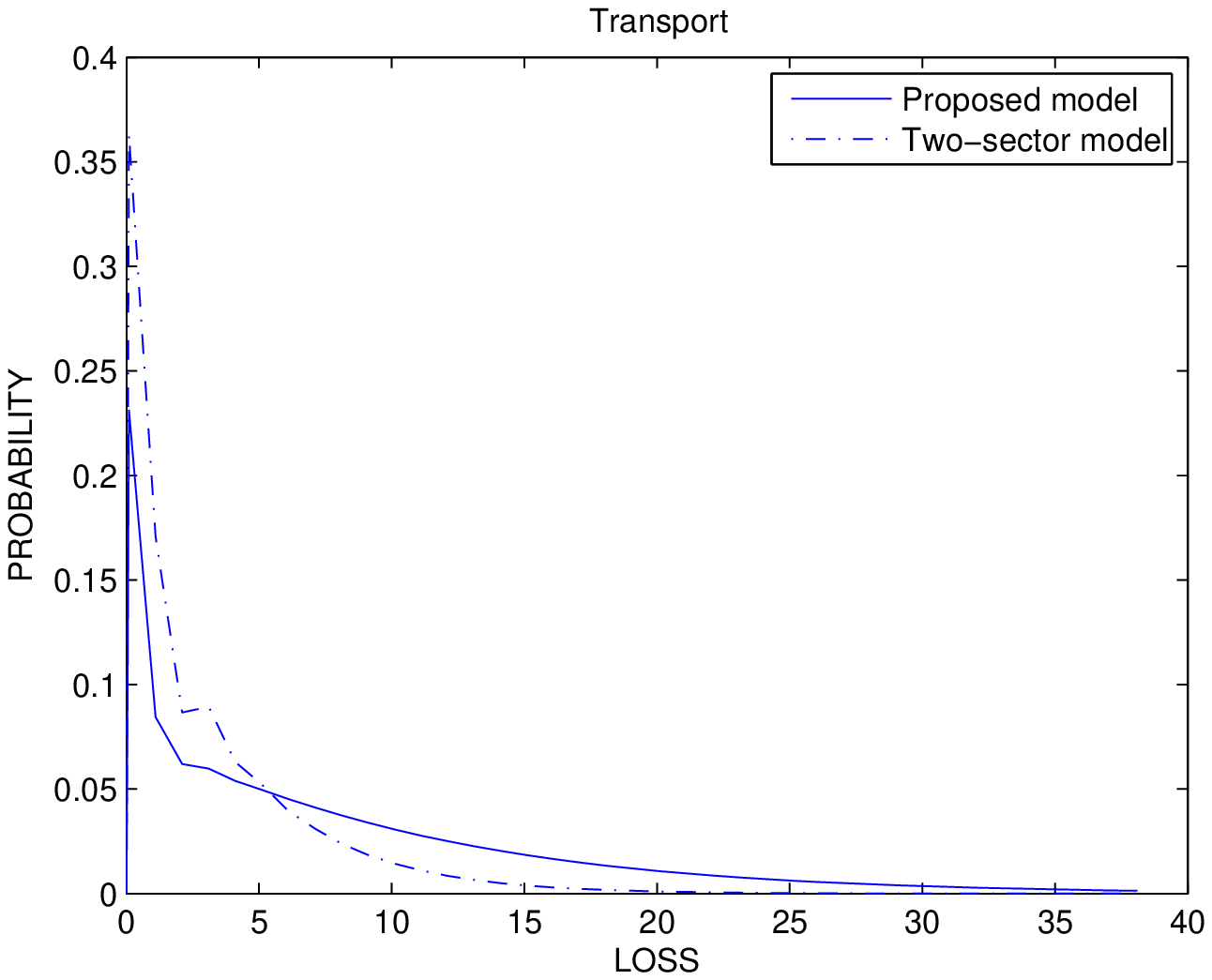}}
	\caption{Loss distribution for proposed model and two-sector model Ching et al. (2010).}
\end{figure}
\vspace{5mm}

From Table 6, we see that for all of the four sectors, the existing two-sector model underestimates
both the CRES and CRVaR.
This reflects that failure to incorporate the contagion effect described in
our proposed model leads to an underestimation of credit risk and has important consequences for
credit risk management, such as inadequate capital charges for credit portfolios. Indeed,
the loss distribution implied by the proposed model has a much fatter tail than that
arising from the existing two-sector model. This explains why the proposed model
provides more prudent estimates for the risk measures than the existing two sector
model via incorporating contagion. We also remark that the contagion model
including the causality of defaults in both direction, (i.e., looping defaults),
has a significant impact on the loss distribution.

\section{A Generalized Model}
As in the basic model, the stochastic process $(X_t, Y_t)$ has the Markov property, where conditioning on $(X_t, Y_t)$, $(X^1_{t+1}, Y^1_{t+1})$ and $(X^2_{t+1}, Y^2_{t+1})$ are stochastically independent.
The joint probability distribution, given the realization of $(X_t, Y_t), \alpha_t, \beta_t$, is given by:
\begin{eqnarray}\label{prob1}
\begin{array}{lll}
p_{(x_t,y_t), (\alpha_t, \beta_t)} (x_{t+1},y_{t+1})  & =&
P\{(X_{t+1}, Y_{t+1}) = (x_{t+1}, y_{t+1}) \mid (X_t, Y_t) = (x_t, y_t),\alpha_t, \beta_t\} \\
&=&  \left(
\begin{array}{c}
x^1_t \\
y^1_{t+1}
\end{array}
\right) {(\alpha}_t)^{y^1_{t+1}} (1-{\alpha}_t)^{x^1_{t+1}}
\times \left(
\begin{array}{c}
x^2_t \\
y^2_{t+1}
\end{array}
\right) {(\beta}_t)^{y^2_{t+1}} (1-{\beta}_t)^{x^2_{t+1}}.\\
\end{array}
\end{eqnarray}
However, instead of maintaining the specific form of a bivariate step function for $\alpha_t, \beta_t$, this model assumes that $\alpha_t$ and $\beta_t$ follow certain Beta distributions depending on $(X_t, Y_t)$.
By assuming a beta density on the unknown transition parameters, the chain becomes a Markov chain with transition matrix containing random parameters. This 
allow us to incorporate parameter uncertainty while, at the same time, retaining  the analytical tractability of the model. For each time period,
the number of defaults has the Beta-binomial distribution depending on the number of defaults in last time period.
The Beta-binomial distribution is extensively used in Bayesian statistics, empirical Bayes methods and classical statistics as an overdispersed binomial distribution. 

Specifically, it is assumed that the density of $\alpha_t$ and $\beta_t$ are given by $f_{\alpha}(x; (X_t, Y_t))$ and $f_{\beta}(x; (X_t, Y_t))$, respectively, and
$$
f_{\alpha}(x; (X_t, Y_t))=\sum_{i=0}^3\frac{h_i(y^1_t, y^2_t)}{B(A_{i1},A_{i2})}x^{A_{i1}-1}(1-x)^{A_{i2}-1}
$$
and
$$
f_{\beta}(x; (X_t, Y_t))=\sum_{i=0}^3\frac{h_i(y^2_t, y^1_t)}{B(B_{i1},B_{i2})}x^{B_{i1}-1}(1-x)^{B_{i2}-1}
$$
where
$$
B(x,y)=\int_0^1 t^{x-1} (1-t)^{y-1} dt
$$
and $A_{ij}, B_{ij}, i=0,1,2,3, j=1,2$ are parameters of the Beta distribution.
From the definition, one can have the following transition probability:

\begin{equation}\label{trans1}
\begin{array}{ll}
&P\{(X^1_{t+1}, Y^1_{t+1})=(x^1_{t+1}, y^1_{t+1}) \mid (X_{t}, Y_{t})=(x_{t}, y_{t})\}\\
=&  \left(
\begin{array}{c}
x^1_t \\
y^1_{t+1}
\end{array}
\right) {E\left[({\alpha}_t)^{y^1_{t+1}} (1-{\alpha}_t)^{x^1_{t+1}} \mid (X_{t}, Y_{t})=(x_{t}, y_{t})\right]}\\
=&  \left(
\begin{array}{c}
x^1_t \\
y^1_{t+1}
\end{array}
\right) \displaystyle \int_0^1 p^{y^1_{t+1}}(1-p)^{x^1_{t+1}} f_{\alpha}(p; (x_t,y_t)) dp\\
=& \left(
\begin{array}{c}
x^1_t \\
y^1_{t+1}
\end{array}
\right){\displaystyle \sum_{i=0}^3}\displaystyle \frac{h_i(y^1_t, y^2_t)}{B(A_{i1}, A_{i2})}\int_0^1 p^{y^1_{t+1}}(1-p)^{x^1_{t+1}} p^{A_{i1}-1}(1-p)^{A_{i2}-1} dp\\
=& \left(
\begin{array}{c}
x^1_t \\
y^1_{t+1}
\end{array}
\right){\displaystyle \sum_{i=0}^3} h_i(y^1_t, y^2_t)
\displaystyle  \frac{B(y^1_{t+1}+A_{i1}, x^1_{t+1}+A_{i2})}{B(A_{i1}, A_{i2})}
\end{array}
\end{equation}

A similar transition probability distribution is shared with the number of defaults in Sector B.
\begin{equation}\label{trans2}
\begin{array}{ll}
&P\{(X^2_{t+1}, Y^2_{t+1})=(x^2_{t+1}, y^2_{t+1}) \mid (X_{t}, Y_{t})=(x_{t}, y_{t})\}\\
=& \left(
\begin{array}{c}
x^2_t \\
y^2_{t+1}
\end{array}
\right){\displaystyle \sum_{i=0}^3} h_i(y^2_t, y^1_t)
\displaystyle \frac{B(y^2_{t+1}+B_{i1}, x^2_{t+1}+B_{i2})}{B(B_{i1}, B_{i2})}
\end{array}
\end{equation}
From the transition probability distribution, to obtain a closed-form solution for a maximum likelihood estimate is difficult, if not impossible. However, we can compute the maximum likelihood estimates using numerical optimization.

\subsection{Default Cycle and Severity}

To derive the joint distribution of $(W^i_{T_i}, T_i)$ for $i=1,2$, we repeat the same steps in Section 2 to compute
$$
P_n(x_1,x_2,h)=P\{T_1 \geq n+1, X^1_n=x_1, X^2_n=x_2, I_{\{Y^2_n>0\}}=h\}.
$$
\begin{lemma}
\begin{equation}\label{lemma21}
\begin{array}{ll}
 P_n(x_1,x_2,0)=& {\displaystyle \sum_{s_1>x_1}{s_1 \choose x_1}} \left[ \displaystyle P_{n-1}(s_1,x_2,0)\frac{B(s_1-x_1+A_{11}, x_1+A_{12})}{B(A_{11}, A_{12})}\frac{B(B_{21}, x_2+B_{22})}{B(B_{21}, B_{22})}\right.\\
& \displaystyle \left.+P_{n-1}{(s_1,x_2,1)}\frac{B(s_1-x_1+A_{31}, x_1+A_{32})}{B(A_{31}, A_{32})}\frac{B(B_{31}, x_2+B_{32})}{B(B_{31}, B_{32})} \right]
\end{array}
\end{equation}
\begin{equation}\label{lemma22}
\begin{array}{ll}
 P_n(x_1,x_2,1)=& {\displaystyle \sum_{s_1>x_1}\sum_{s_1>x_1}{s_1 \choose x_1} {s_2 \choose x_2}}\left[ P_{n-1}(s_1,s_2,0)
 \frac{B(s_1-x_1+A_{11}, x_1+A_{12})}{B(A_{11}, A_{12})} 
 \frac{B(s_2-x_2+B_{21}, x_2+B_{22})}{B(B_{21}, B_{22})}\right.\\
&\displaystyle \left.+P_{n-1}(s_1,s_2,1)\frac{B(s_1-x_1+A_{31}, x_1+A_{32})}{B(A_{31}, A_{32})}\frac{B(s_2-x_2+B_{31}, x_2+B_{32})}{B(B_{31}, B_{32})} \right]
\end{array}
\end{equation}
where the initial condition is given by
$$
P_0(x_1,x_2,h)=\left\{
\begin{array}{ll}
1, & (x_1,x_2,h)=(x^1_0,x^2_0,1)\\
0, & otherwise
\end{array}
\right.
$$
\end{lemma}

\begin{proof}
We prove the first equality. The proof of the second one is similar.
As in the proof of Lemma 1,
$$
\begin{array}{lll}
 & P_n(x_1,x_2,0)\\
=&{\displaystyle \sum_{s_1>x_1}\sum_{h=0,1}} P_{n-1}(s_1,x_2,h)\\
 & \times P\{Y^1_n>0, X^1_n=x_1, X^2_n=x_2, I_{\{Y^2_n >0\}}=0 \mid Y^1_{n-1}>0, X^1_{n-1}=s_1, X^2_{n-1}=x_2, I_{\{Y^2_{n-1} >0\}}=h\}
\end{array}
$$
Note that by (\ref{trans1}) and (\ref{trans2}),
$$
\begin{array}{ll}
& P\{Y^1_n>0, X^1_n=x_1, X^2_n=x_2, I_{\{Y^2_n >0\}}=0 \mid Y^1_{n-1}>0, X^1_{n-1}=s_1, X^2_{n-1}=x_2, I_{\{Y^2_{n-1} >0\}}=0\}\\
=& \displaystyle {s_1 \choose x_1}\frac{B(s_1-x_1+A_{11}, x_1+A_{12})}{B(A_{11}, A_{12})}\frac{B(B_{21}, x_2+B_{22})}{B(B_{21}, B_{22})}
\end{array}
$$
and
$$
\begin{array}{ll}
& P\{Y^1_n>0, X^1_n=x_1, X^2_n=x_2, I_{\{Y^2_n >0\}}=0 \mid Y^1_{n-1}>0, X^1_{n-1}=s_1, X^2_{n-1}=x_2, I_{\{Y^2_{n-1} >0\}}=1\}\\
=&  \displaystyle {s_1 \choose x_1}\frac{B(s_1-x_1+A_{31}, x_1+A_{32})}{B(A_{31}, A_{32})}\frac{B(B_{31}, x_2+B_{32})}{B(B_{31}, B_{32})}
\end{array}
$$
Combining these two results, (\ref{lemma21}) follows.
\end{proof}

Hence the joint distribution of $(T_1, W^1_{T_1})$ follows
\begin{proposition}
\begin{equation}\label{proposition2}
\begin{array}{ll}
P\{(T_1, W^1_{T_1})=(n,x))\}=&\displaystyle \sum_{x_2} P_{n-1}(x^1_0-x,x_2,0)\frac{B(A_{11}, x^1_0-x+A_{12})}{B(A_{11},A_{12})}\\
& \displaystyle +\sum_{x_2} P_{n-1}(x^1_0-x,x_2,1)\frac{B(A_{31}, x^1_0-x+A_{32})}{B(A_{31},A_{32})}
\end{array}
\end{equation}
\end{proposition}

\begin{proof}
As in the proof of Proposition 1,
$$
\begin{array}{lll}
 &P\{(T_1,W^1_{T_1})=(n,x)\}\\
=& {\displaystyle \sum_{x_2}\sum_{h=0,1}}P_{n-1}(x^1_0-x,x_2,h)\\
& \times P\{ Y^1_n=0, X^1_n=x^1_0-x \mid Y^1_{n-1}>0, X^1_{n-1}=x^1_0-x, X^2_{n-1}=x_2, I_{\{Y^2_{n-1} >0\}}=h\}
\end{array}
$$
Note that by (\ref{trans1}),
$$
\begin{array}{ll}
& P\{ Y^1_n=0, X^1_n=x^1_0-x \mid Y^1_{n-1}>0, X^1_{n-1}=x^1_0-x, X^2_{n-1}=x_2, I_{\{Y^2_{n-1} >0\}}=0\}\\
=& \displaystyle \frac{B(A_{11}, x^1_0-x+A_{12})}{B(A_{11}, A_{12})}
\end{array}
$$
and
$$
\begin{array}{ll}
& P\{ Y^1_n=0, X^1_n=x^1_0-x \mid Y^1_{n-1}>0, X^1_{n-1}=x^1_0-x, X^2_{n-1}=x_2, I_{\{Y^2_{n-1} >0\}}=1\}\\
=&  \displaystyle\frac{B(A_{31}, x^1_0-x+A_{32})}{B(A_{31}, A_{32})}
\end{array}
$$
Combining these two, (\ref{proposition2}) follows.
\end{proof}

\section{Concluding Remarks}

We propose a two-sector Markovian infectious model.
The proposed model incorporated two important features of
credit contagion, namely, the chain reactions of defaults
and the bi-lateral causality of defaults between two
industrial sectors.
We capture the chain reactions of defaults by postulating that the future default
probability switches over time according to the
current number of defaults of two industrial sectors.
The bi-lateral causality of defaults
meant that defaults in one sector are caused
by defaults in another sector, and vice versa.
This bi-lateral causality of defaults enriches
the dependent structures of credit risk model.
We provide an efficient estimation method
of the proposed model based on the maximum likelihood
estimation. Two important risk measures, namely,
the CRVaR and the CRES, were evaluated under the proposed model.
{ To provide a more flexible and realistic modeling framework
for the dynamics of default probabilities,
we extend the model to a case where default probabilities are
Beta random variables given the realization of the state in the
previous time period.}

We also conduct empirical studies on the credit risk
models using real default data. We adopted the BIC
to compare the proposed model with the existing two-sector model
proposed in Ching et al. (2010).
The numerical results reveal that the proposed two-sector model outperforms
empirically the existing model.
By comparing the risk measures evaluated from the proposed model
and those evaluated from the existing two-sector
model, we found that failure to incorporate
the contagion effect described in the proposed
model leads to an underestimation of risk measures.
This provides some evidence to support the proposed model.

One possible topic for future research may be to incorporate
the impact of the number of defaults on the likelihood
of future defaults via a different parametrization of
the future default probability.
In current paper, we assumed that the joint future default probability
switches over time depending on the region where
the current number of defaults falls in.
Four parameters, namely, $a_0$, $a_1$, $a_2$ and $a_3$ were involved.
To provide a more parsimonious way to incorporate
the current number of defaults on the joint
future default probability, one may consider the
following parametrization for the future default
probability:
\begin{eqnarray*}
\alpha_t = a_0 + a_1 y^1_t + a_2 y^2_t \ ,
\end{eqnarray*}
where $y^1_t$ and $y^2_t$ are the current numbers of defaults
in the two industrial sectors. Using this parametrization,
we can reduce the number of parameters by one and accounts
for more information of the current number of defaults
when evaluating the future default probability.

\vspace{5mm}
\noindent {\bf Acknowledgment}:
The authors would like to thank the anonymous referees for their helpful
comments. Research supported in part by RGC Grants 7017/07P, HKU CRCG Grants
and HKU Strategic Research Theme Fund on
Computational Physics and Numerical Methods.

\end{document}